\documentclass[11pt, letterpaper]{article}
\usepackage{graphicx} %

\usepackage{fullpage}
\usepackage{color,xspace}
\usepackage{amsmath,amssymb,amsthm,enumerate}
\usepackage{thm-restate}

\usepackage{graphicx}
\usepackage{subcaption}
\usepackage{url}
\usepackage{hyperref}
\usepackage[capitalise]{cleveref}
\hypersetup{colorlinks=true}
\hypersetup{final}
\hypersetup{linkcolor=black,anchorcolor=black,citecolor=black,urlcolor=black}
\usepackage{algorithm, algpseudocode}
\usepackage{ifdraft}
\usepackage[bibliography=common]{apxproof}
\newtheorem{theorem}{Theorem}[section]
\newtheorem{definition}[theorem]{Definition}

\newtheorem{prop}[theorem]{Proposition}

\newtheorem{lemma}[theorem]{Lemma}

\newtheorem{claim}[theorem]{Claim}

\theoremstyle{plain}
\newtheoremrep{claimapp}[theorem]{Claim}
\theoremstyle{plain}
\newtheoremrep{theoremapp}[theorem]{Theorem}
\theoremstyle{plain}
\newtheoremrep{lemmaapp}[theorem]{Lemma}
\theoremstyle{plain}
\newtheoremrep{propositionapp}[theorem]{Proposition}
\theoremstyle{plain}
\newtheoremrep{corollaryapp}[theorem]{Corollary}

\newenvironment{proofof}[1]{

\noindent{\bf Proof of {#1}:}}
{\hfill\qed

}

\DeclareMathOperator{\ALG}{ALG}
\DeclareMathOperator{\OPT}{OPT}

\newcommand{\services}{\Lambda}
\newcommand{\pservices}{\Lambda_1} %
\newcommand{\cservices}{\Lambda^{\circ}} %
\newcommand{\ciservices}{\Lambda^{\circ}_i} %
\newcommand{\tiservices}{\hat{\Lambda}_i}
\newcommand{\predd}{\tilde{d}} %
\DeclareMathOperator{\inv}{inv}
\DeclareMathOperator{\disj}{disj}
\newcommand{\numreqinv}{\textsc{\#ReqInv}}
\newcommand{\numiteminv}{\ensuremath{\eta}}
\newcommand{\Prem}{\hat{Q}}

\newcommand{\premlambda}{\hat{\Lambda}}
\newcommand{\local}{\ensuremath{\textsf{Local-Greedy}}\xspace}
\newcommand{\greedy}{\ensuremath{\textsf{Classic-Greedy}}\xspace}
\newcommand{\A}{\ensuremath{\textsf{A}}\xspace}
\newcommand{\mgreedy}{\ensuremath{\textsf{Folklore-Greedy}}\xspace}

\usepackage[colorinlistoftodos,textsize=small,color=red!25!white,obeyDraft]{todonotes}
\ifdraft{
\newcommand{\mdnote}[1]{\todo[backgroundcolor=red!25,bordercolor=red!50!black]{\textbf{MD:} #1}}
\newcommand{\mdnoteinline}[1]{\todo[inline, size=\normalsize, color=red!20]{Mike's Note: #1}}
\newcommand{\wunoteinline}[1]{\todo[inline, size=\normalsize, color=red!20]{William's Note: #1}}
\newcommand{\wunote}[1]{\todo[backgroundcolor=red!25,bordercolor=red!50!black]{\textbf{WU:} #1}}
\newcommand{\jfnote}[1]{\todo[backgroundcolor=red!25,bordercolor=red!50!black]{\textbf{JF:} #1}}
\newcommand{\jfnoteinline}[1]{\todo[inline, size=\normalsize, color=red!20]{Jeremy's Note: #1}}
}
{
\newcommand{\mdnote}[1]{}
\newcommand{\mdnoteinline}[1]{}
\newcommand{\wunoteinline}[1]{}
\newcommand{\wunote}[1]{}
\newcommand{\jfnote}[1]{}
\newcommand{\jfnoteinline}[1]{}
}

\newcommand{\defn}[1]{\textbf{\emph{#1}}}

\title{Learning-Augmented Online Algorithms for Nonclairvoyant Joint Replenishment Problem with Deadlines}
\author{Michael Dinitz\thanks{Johns Hopkins University. Email: \url{mdinitz@cs.jhu.edu}.  Supported in part by NSF award 2228995.} \and 
Jeremy T. Fineman\thanks{Georgetown University. Email: \url{jf474@georgetown.edu}. Supported in part by NSF awards 2106759 and 1918989.} \and 
Seeun William Umboh\thanks{The University of Melbourne, and ARC Training Centre in Optimisation Technologies, Integrated Methodologies, and Applications (OPTIMA). Email: \url{william.umboh@unimelb.edu.au}. Supported by the Australian Government through the Australian Research Council
DP240101353.}}

\begin{document}

\begin{titlepage}
\def\thepage{}
\thispagestyle{empty}
\maketitle

\setcounter{page}{0}
\begin{abstract}
  This paper considers using predictions in the context of the online Joint Replenishment Problem with Deadlines (JRP-D).  Prior work includes asymptotically optimal competitive ratios of $O(1)$ for the clairvoyant setting and $O(\sqrt{n})$ of the nonclairvoyant setting, where $n$ is the number of items.  The goal of this paper is to significantly reduce the competitive ratio for the nonclairvoyant case by leveraging predictions: when a request arrives, the true deadline of the request is not revealed, but the algorithm is given a predicted deadline.

  The main result is an algorithm whose competitive ratio is $O(\min(\numiteminv^{1/3}\log^{2/3}(n), \sqrt{\numiteminv}, \sqrt{n}))$, where $n$ is the number of item types and $\eta\leq n^2$ quantifies how flawed the predictions are in terms of the number of ``instantaneous item inversions.''  Thus, the algorithm is robust, i.e., it is never worse than the nonclairvoyant solution, and it is consistent, i.e., if the predictions exhibit no inversions, then the algorithm behaves similarly to the clairvoyant algorithm. Moreover, if the error is not too large, specifically $\eta < o(n^{3/2}/\log^2(n))$, then the algorithm obtains an asymptotically better competitive ratio than the nonclairvoyant algorithm. We also show that all deterministic algorithms falling in a certain reasonable class of algorithms have a competitive ratio of $\Omega(\eta^{1/3})$, so this algorithm is nearly the best possible with respect to this error metric.
  \end{abstract}
\end{titlepage}
\clearpage

\section{Introduction}
\label{sec:intro}

In many practical online optimization problems, we are not required to serve a request immediately upon arrival but are allowed to either postpone servicing the request either up to a deadline or by paying a delay cost that depends on the waiting time between arrival time and service time. Moreover, at any point in time, the algorithm can aggregate a subset of available requests and serve them together as a batch, leveraging economies of scale. These problems are called \emph{online problems with deadlines or delay}. An important example of such a problem is the Joint Replenishment Problem, which we describe next.

The Joint Replenishment Problem (JRP) is a cornerstone problem in Operations Research and plays an important role in supply chain management. At a high level, the problem consists of a sequence of requests that arrive over time where each request is associated with one of $n$ item types (also called ``items"). %
At any time, we may choose to serve a set of requests $Q$ and incur a fixed cost, called the \defn{joint ordering cost}, plus a marginal cost, called the \defn{item ordering cost}, for each item type associated with a request in $Q$. (Note that the marginal cost is per item type, not request, so servicing additional requests of the same item type is effectively free.) In the general version of the problem, each request also has a (possibly different) delay cost that depends on how long it waited between its arrival and service times. This paper considers the well-studied and important special case called JRP with deadlines (JRP-D): here instead of a delay cost, each request has a deadline and must be served no later than its deadline. 

JRP has been studied in both the offline and online settings. In the offline setting, the problem is known to be APX-hard \cite{jrp-apx-hard, BienkowskiBCDNSSY15} and admits small constant-factor approximations (c.f.~\cite{BienkowskiBCJNS14}). This paper focuses on the online setting: at each time $t$, we only know of requests whose arrival times are not later than $t$, and the decision is which (if any) items to serve at time $t$.

\paragraph{Nonclairvoyance vs Clairvoyance.} For online problems with deadlines or delay, a standard distinction is made between the clairvoyant and nonclairvoyant settings.  In the \defn{clairvoyant} setting, when a request arrives, we also receive its deadline or its delay cost function for all future times. Online JRP was first studied in the clairvoyant setting by Buchbinder et al.~\cite{online-jrp-delay} who gave a deterministic  primal-dual algorithm that is $3$-competitive for the general delay version and $2$-competitive for JRP-D. Later, Bienkowski et al.~\cite{BienkowskiBCJNS14} gave a deterministic greedy $2$-competitive algorithm for JRP-D and proved a matching lower bound. %

While most previous work on online problems with deadlines or delay has focused on the clairvoyant setting, there has been recent interest in studying the \defn{nonclairvoyant} setting. In this setting, at each time $t$, for every request that has arrived but not yet been served, the algorithm only knows of the marginal increase in the delay cost incurred by the request if it is not served at the current time $t$. For online problems with deadlines, nonclairvoyance translates into knowing only if the deadline of a request is at the current time $t$, in which case it must be served immediately, or later than $t$. It turns out that the nonclairvoyant setting is significantly harder. For JRP, Le et al.~\cite{LeUX23} gave a deterministic $O(\sqrt{n})$-competitive algorithm. Moreover, for JRP-D and hence also JRP, this is best-possible even for randomized algorithms~\cite{AzarGGP21,LeUX23}. Beyond JRP, there has been work on nonclairvoyant algorithms for generalizations of JRP \cite{azar2020set,Touitou23,Ezra0PRU24}, as well as online matching with delay \cite{DeryckereU23} and online service with delay \cite{Touitou25}.

\paragraph{Using Predictions to Overcome Nonclairvoyance.}
There is an enormous gap between the clairvoyant setting (where we have a $2$-competitive algorithm) and the nonclairvoyant setting (where the best possible competitive ratio is $\Omega(\sqrt{n})$).  The goal of this paper is to circumvent the difficulty of the nonclairvoyant setting through the use of learning-augmented algorithms, also known as ``algorithms with predictions''.  

In the algorithms with predictions framework, in addition to the regular input the algorithm is also given ``predictions'' of some sort (presumably the output of some machine learning algorithm or system). The goal is to use the predictions to get beyond traditional worst-case bounds: if the predictions are ``good'' then we would like our algorithm to perform extremely well, while if the prediction are ``bad'' then we would still like to have traditional worst-case guarantees.   In other words, we want the best of both worlds: good performance in the average case thanks to machine learning, but also robustness and good performance in the
worst-case from traditional algorithms.  

While similar ideas have appeared in many places in the past, the formal study of algorithms with predictions was pioneered by Lykouris and Vassilvitskii~\cite{lykouris2021competitive}.  While this framework can apply to many different algorithmic settings (and has been applied to settings like combinatorial algorithms \cite{DinitzILMV21,davies2023predictive}, differential privacy \cite{abs-2210-11222}, data structures \cite{DILMNV24,lin2022learning,VaidyaKMK21} and mechanism design \cite{agrawal2022learning}, to name just a few), the focus of~\cite{lykouris2021competitive} was on online caching, and the vast majority of papers in this area continue this line of work by considering online problems.  We note that not only are predictions useful (as in caching) to allow the online algorithm to perform more like the offline, but they have also been used to allow nonclairvoyant algorithms to perform more like clairvoyant algorithms (see, e.g., \cite{IKMP21,permutation}).  So we build on this line of work by studying nonclairvoyant JRP with predictions.

\paragraph{Predictions for JRP.}
The first question for any problem in the algorithms with predictions framework is obvious: what should the prediction be?  We make a natural choice: when a request $q$ arrives at time $a_q$, because we are in the nonclairvoyant setting we are not given its true deadline $d_q$, but are instead given a \defn{predicted deadline} $\predd_q$.  We note that while sometimes using non-obvious predictions is a useful tool that can lead to improved performance (see, e.g.,~\cite{LattanziLMV}), it is standard and common (particularly in online algorithms) for the prediction to be essentially the ``missing data'', i.e., whatever information is missing that would allow for improved performance.  So our choice of deadline as the prediction is the obvious and standard choice.

We note that another obvious choice would be an \emph{ordering} of requests with respect to their predicted deadlines.  Permutation predictions have been used in the past for similar problems, most notably by \cite{permutation} in the context of nonclairvoyant scheduling.  It turns out that, thanks to the error metric that we adopt (see following discussion), the way we use predicted deadlines is essentially equivalent to having a permutation prediction.  

\subsection*{Error Metric} Once we have decided on what the prediction is, the next question in the framework is to fix an \defn{error metric}.  One extreme is the ``consistency vs robustness'' framework, where we design an algorithm with two properties.  First, it does well if the predictions are perfect ($\predd_q = d_q$ for all requests $q$), and in particular does approximately as well as would be possible in the clairvoyant case; this is called a \defn{consistency} guarantee.  Second, if the predictions are not accurate, then the worst-case guarantee is not much worse than the traditional nonclairvoyant guarantee; this is called a \defn{robustness} guarantee.  While much of the initial work on algorithms with predictions studied the tradeoff between consistency and robustness, more recent research aims to give refined bounds that quantify the performance of the algorithm as a function of \emph{how inaccurate} the predictions are (see the excellent list of papers in this area at~\cite{ALPSweb}).  In particular, the goal is to obtain a competitive ratio that is a function of both the problem instance (notably $n$, the number of items here) and the error in the predictions. If the predictions are perfectly accurate we get back the consistency guarantee, but if they are ``close'' to being accurate then we should still get nontrivial guarantees that are much better than the completely worst-case robustness guarantee.  

\paragraph{Numeric.} So we need a way to quantify how close the predictions $\{\predd_{q}\}$ are from the true deadline $\{d_q\}$, and we want to design an algorithm whose performance degrades gracefully as a function of this distance.  In the case of JRP-D, the choice of this error metric turns out to be not quite as obvious for us as the choice of prediction was.  A natural first try is to consider these predictions as a vector (indexed by request $q$) and use a traditional $\ell_p$ vector norm, e.g., the maximum deadline error $\|d - \predd\|_{\infty}$, the total deadline error $\|d - \predd\|_1$, or some intermediate norm like $\|d - \predd\|_2$.  Unfortunately, for nonclairvoyant JRP these norms are somewhat unsatisfactory.  In particular, they are not scale-invariant: if we simply change the time scale on which we operate, e.g., by multiplying all times by some constant factor $c$, then any of these norms will also increase by a factor of $c$.  But we have not really made the predictions any more or less accurate; we have simply changed the timescale.  Moreover, it is not hard to see that the important information for the clairvoyant case is not the \emph{value} of the deadlines, but rather their \emph{ordering}.  For example, suppose the predicted deadlines are all incorrect but their order is correct. Then the standard $2$-competitive clairvoyant algorithm would still be $2$-competitive---in fact, all known algorithms for JRP-D and its generalizations are invariant under order-preserving perturbations of deadlines---but any error metric based on norms could be large. So we would intuitively like our error metric to be a function of the induced ordering, rather than the values themselves.

\paragraph{Inversions.} Clearly any collection of predicted or true deadlines gives an ordering of the requests.  A natural notion of error then is the number of \defn{inversions} between these orderings; that is, the number of pairs of requests $q, q'$ where $q$ is before $q'$ in the true ordering but $q'$ is before $q$ in the predicted ordering.  Unfortunately, this notion of error has a similar issue to the scale-invariant issue earlier: the number of requests can be arbitrarily larger than the number $n$ of item types.  So, for example, given a sequence of requests with some number of inversions between the true deadlines and the predicted deadlines, if we simply \emph{repeat} this sequence again we end up with double the error, but without any meaningful change in the problem instance.  Moreover, since there exists a worst-case nonclairvoyant solution that is $O(\sqrt{n})$ competitive, our error metric should ideally be upper bounded by a function of~$n$, not the arbitrarily large number of requests.  

But the bad example of repeating a sequence gives a natural notion of error that fixes these problems: \defn{item inversions}.  We say that two \emph{items} are inverted if there exists two requests for these items that are inverted.  More formally, items $i$ and $j$ are inverted if there exist requests $q$ to item $i$ and $q'$ to item $j$ such that $d_q < d_{q'}$ but $\predd_q > \predd_{q'}$.  This is one of the main notions of inversions that we will use, and in fact our lower bound (\cref{thm:lb-main}) is in terms of item inversions.

However, item inversions have another issue: it is trivial to design instances where there are many item inversions, but at any point in time there are very few actual request inversions.  In such a scenario we would like to think that our predictions are actually quite good.  So we arrive at our final error metric: \defn{instantaneous item inversions}.  We define this formally in \cref{sec:prelim}, but it is essentially the maximum over all time points $t$ of the number of item inversions for requests whose interval includes $t$.  Let $\numiteminv$ denote the total number of instantaneous item inversions.  To state performance bounds more concisely (they have a multiplicative dependence on item inversions), we define $\numiteminv$ to be 1 when there are zero inversions.  Then we now have our final question: \emph{can we design an algorithm with good performance as a function of $\numiteminv$?}

\subsection{Our Results}

Our main technical result is an algorithm for JRP-D with predictions with good performance as a function of $\numiteminv$. 

\begin{restatable}{theorem}{maintheorem} \label{thm:main}
    There is an algorithm for nonclairvoyant JRP-D with predictions with competitive ratio at most  $O(\numiteminv^{1/3} \log^{2/3} n)$.
\end{restatable}

This theorem implies that we start getting nontrivial bounds (better than the $O(\sqrt{n})$-competitive nonclairvoyant algorithm) when $\numiteminv \leq O(n^{3/2} / \log^2 n)$. %

We also show (\cref{sec:combine}) that as long as an algorithm for JRP-D satisfies a few key properties, it can be combined with other such algorithms to achieve  the \emph{best} of the competitive ratios offered by any of the algorithms.  So we give two more algorithms and bounds to improve extreme regimes where Theorem~\ref{thm:main} is weak.  First, we are able to prove a $O(\sqrt{\numiteminv})$-competitive ratio for a slight variant of our new algorithm from Theorem~\ref{thm:main}, which is better when $\numiteminv$ is very small but worse when $\numiteminv$ is extremely large.  Second, while the $O(\sqrt{n})$-competitive nonclairvoyant algorithm of~\cite{LeUX23} does \emph{not} satisfy the properties we need for combinability, we show that it can be modified to have these properties without affecting its competitive ratio.  This gives an improved bound when $\numiteminv$ is  large.  Putting these together, we obtain the following theorem:

\begin{restatable}{theorem}{combinedtheorem} \label{thm:maincombined} 
    There is an algorithm for nonclairvoyant JRP-D with predictions with competitive ratio at most $O(\min(\numiteminv^{1/3}\log^{2/3} n, \sqrt{\numiteminv}, \sqrt{n})$.
\end{restatable}    

Our final algorithm is thus both consistent and robust. Moreover its competitive ratio grows relatively slowly as a function of the error. We remark that all of our algorithms are deterministic.

\subsubsection{Techniques and Approach}
To understand the strength of this result and the techniques we use to get it, consider an extreme scenario of ``useless'' predictions.  The worst-case number of instantaneous item inversions is $\numiteminv = \Theta(n^2)$.  Moreover, if the predictions are completely uninformative then we are essentially back in the traditional nonclairvoyant setting without predictions, where an $O(\sqrt{n})$ competitive ratio is possible.  Therefore the best competitive ratio as a function of $\numiteminv$ that we can hope to achieve is  $\Theta(\sqrt{n}) =  \Theta(\numiteminv^{1/4})$.  So a natural goal for us is to aim for a competitive ratio of $O(\numiteminv^{1/4})$.  

With this as our goal, a natural approach (as in all algorithms with predictions settings) is to see what happens if we simply run the state of the art greedy $2$-competitive clairvoyant algorithm~\cite{BienkowskiBCJNS14} directly on the predictions.  We call this algorithm \greedy and describe it formally in \cref{sec:greedy}, but informally it works as follows.  When the deadline of a pending request is reached at some time $t$, we sort the requests by (predicted) deadline and keep adding items to our set in this order and stops when adding the next item causes the total of item ordering costs of our set to exceed the joint ordering cost; in other words, we choose the largest prefix of items whose total item ordering cost does not exceed the joint ordering cost.  We then serve all pending requests for these items. 

\paragraph{\greedy fails.} 

In \cref{sec:greedy} we modify \greedy slightly to make it easier to analyze albeit sacrificing constants---instead of choosing the smallest prefix of items whose total item ordering cost does not exceed the joint ordering cost---and call it \mgreedy. We show that the competitive ratio of \mgreedy is at most $O(\numiteminv)$, so we get an improvement over the nonclairvoyant algorithm when $\numiteminv \leq O(\sqrt{n})$. This is a nontrivial result and so we include it for completeness since the greedy algorithm is of fundamental interest; in particular, variants of \greedy and \mgreedy are subroutines in generalizations of JRP. Unfortunately, there is a fundamental roadblock: neither \greedy nor \mgreedy beats $\Omega(\sqrt{\numiteminv})$, and hence fall far from our goal.

\begin{lemma}
 \greedy and \mgreedy has competitive ratio $\Omega(\sqrt{\numiteminv})$.
\end{lemma}
\begin{proof}
Create $2k$ items $b_1,\ldots,b_k$ (the \emph{black items}) and $r_1,\ldots,r_k$ (the \emph{red items}).  Set the joint ordering cost at 1 and all item costs at $1/k$, meaning the greedy algorithm would always try to service $k$ additional items on any deadline (if possible). 
For each red item $r_i$, create a single request with release time $0$, deadline $2i$, and predicted deadline $2i$.  For each black item $b_i$, create $k$ requests: the $j$'th request has arrival time $2j-1$, deadline $3k$, and predicted deadline $2j+1$. (Unlike the red items, the black requests do not depend on the item number~$i$.)  Note that every red item is inverted with every black item, so $\numiteminv = \Theta(k^2)$.  Since all requests for red items are available at time $0$ and all requests for black items are available at time $3k$, clearly we can just use two services (one at time $0$ and one at time $3k$), each of which services $k$ items with service cost $1+k\cdot(1/k)=2$.  So $\OPT \leq 4$.  

We now consider \greedy inductively. At each even time step $2i$, \greedy encounters a deadline for the lone request on $r_i$.  At this point there are outstanding requests for red items $r_{i+1},r_{i+2},\ldots,r_k$ as well as requests for all of the black items (notably those released at time $2i-1$). But each of the black items has a request  with predicted deadline $2i+1$, whereas the next red item's request has a (predicted and true) deadline of $2i+2$. Thus \greedy prioritizes the black items and services $r_i$ and all of the black items at time $2i$.  
Thus \greedy has cost $\Omega(k)$, and so has competitive ratio $\Omega(k) = \Omega(\sqrt{\numiteminv})$. The same argument works for \mgreedy as well since both algorithms behave exactly the same on this instance.
\end{proof}

\paragraph{Improved algorithm.}
To beat a competitive ratio of $\sqrt{\numiteminv}$, we thus need a new algorithm.  We design the algorithm \local{} to get around the above example.  While the exact algorithm is somewhat complex (we give it formally in \cref{sec:local-greedy}), intuitively it takes inspiration from the classical Marking algorithm for caching~\cite{marking} and exploits the fact that $\OPT$ is lower bounded by the joint ordering cost times the number of requests with disjoint intervals.  The algorithm explicitly keeps track of phases, where a new phase starts when a deadline hits for a request with interval that is disjoint from the interval that started the phase.  And then in each phase we essentially run the greedy algorithm, but only on requests that were active at the beginning of the phase.  So this algorithm is local in the sense that when a phase begins, we fix the set of requests that we might consider serving anytime in the phase, and any new requests that get released are deferred until the next phase.    

It is not hard to see that \local{} is actually optimal on our previous bad instance: the first phase starts at time $0$ and includes only the red items, so at the first deadline the algorithm will service all of the red items and none of the black items.  And since all red items have been serviced the next deadline that hits is the common deadline of all black requests, so it will service all the black items at that time.  Unfortunately, it turns out that we are still in trouble: we can construct another bad example on which \local{} also has competitive ratio $\Omega(\sqrt{\numiteminv})$.  This is a little more complex, so we present this instance in \cref{sec:local-greedy-tight}.

We can in fact show a matching upper bound: the competitive ratio of \local{} on \emph{every} instance is at most $O(\sqrt{\numiteminv})$ (see \cref{sec:local-greedy-non-uniform}).  This is a useful bound when $\numiteminv$ is extremely small, but it is still far from our goal.  To improve on this, we note an interesting feature of the bad example: unlike the first example (showing a lower bound on the standard greedy algorithm), in this instance different items have different item ordering costs.  If we are willing to lose an $O(\log n)$ factor, we can bucket the items into weight classes so that all items in a class have roughly the same ordering cost (up to a constant factor). It is not hard to see that this basically resolves the above bad example we had for \local{}.  Our main contribution is to show that this actually implies a better bound for \local{} in general: we show that when all items have the same ordering costs, \local{} is actually $O(\numiteminv^{1/3})$-competitive.  This fact coupled with the bucketing  then implies \cref{thm:main}.

\subsubsection{Lower bound} 

\cref{thm:main} unfortunately does not achieve our original goal, of getting competitive ratio of $O(\numiteminv^{1/4})$.  But we show that this is for a good reason: while $\Omega(\numiteminv^{1/4})$ is a natural lower bound based on completely uninformative predictions, we prove a stronger lower bound (nearly matching our upper bound) of $\Omega(\numiteminv^{1/3})$ against a class of algorithms which we call ``semi-memoryless'' in the ``limited information model''.  Formal definitions of these terms can be found in \cref{sec:lb}, but informally, in the limited information model we make the assumption that if we service a request strictly before its deadline then we do not actually find out when its deadline would have been, and we say that an algorithm is semi-memoryless if its memory is empty whenever there are no outstanding requests.  Semi-memorylessness is a very natural property that all previous algorithms for JRP have had (and which \local{} has), so this is not a major assumption.  Similarly, for most applications it is completely reasonable to assume that servicing a request means that we will never find out its actual deadline since it will never trigger a required service.

So under these two assumptions, we get our nearly-matching lower bound.  We note that in the following theorem we can even take $\numiteminv$ to be the number of item inversions rather than instantaneous item inversions, giving a stronger lower bound.

\begin{restatable}{theorem}{mainlb} \label{thm:lb-main}
For every constant $c \geq 1$ and for a sufficiently large number of item types, every deterministic semi-memoryless algorithm in the limited information model has competitive ratio at least $\max\left(\frac{\numiteminv^{1/3}}{c^{1/3}(c+2)}, c\right)$. 
\end{restatable}

\jfnoteinline{Pretty sure it would work for any $c \leq \sqrt{n}$. Not sure if that makes anything more compelling, so leaving as is.}

\cref{thm:lb-main} provides a construction that sets deadlines one of two possible ways depending on how the algorithm chooses to service requests.  In one case, there are $\eta = \Theta(n^{3/2})$ inversions, but this case is chosen only if it would cause the algorithm to have competitive ratio at least $\Omega(\sqrt{n}) = \Omega(\eta^{1/3})$. In the other case, there are no inversions ($\eta = 1$), but the algorithm has competitive ratio at least~$c$.  Because the theorem holds for all constants $c \geq 1$, that effectively means that either the algorithm fails to be constant competitive when $\eta = 1$, or it is $\Omega(\eta^{1/3})$ competitive when $\eta = \Theta(n^{3/2})$

We note that \cref{thm:lb-main} uses an instance where all item ordering costs are the same, for which our \local{} is also $O(\eta^{1/3})$ competitive. For uniform weights, \local{} is thus effectively optimal with respect to the error metric of (instantaneous) item inversions. 

\subsection{Paper Outline}
We begin in \cref{sec:prelim} with preliminaries and definitions.  In \cref{sec:local-greedy} we define the \local{} algorithm and give our analysis of it, proving \cref{thm:main} as well as a weaker $O(\numiteminv^{1/2})$-competitive bound along the way. In \cref{sec:combine} we show how to combine different algorithms to get ``best of all worlds" guarantee. Then in \cref{sec:lb} we prove our main lower bound, \cref{thm:lb-main}.  In \cref{sec:greedy} we study the modified greedy algorithm \mgreedy (even though it gives weaker bounds) since it is a natural and simple algorithm which is well-studied in the clairvoyant setting.  %

\section{Preliminaries and Notation} \label{sec:prelim}
In the Joint Replenishment Problem with Deadlines (JRP-D), there is a set of $n$ items,
each with an \emph{item ordering cost} $w_i$, and a \emph{joint ordering cost}
$w_0 \geq w_i$ for every item $i$. Then, $m$ requests arrive over time; each
request $q$ has an associated item $v_q$, arrival time $a_q$, and deadline
$d_q$. For a set of requests $Q'$, we use $v(Q')$ to denote the set of items
associated with $Q'$. Serving a set $Q'$ of requests at time $t$ at cost $c(Q')
:= w_0 + \sum_{i \in v(Q')} w_i$. Note that item $i$ contributes $w_i$ to the
cost if there is at least one request in $Q'$ on item $i$. We use $\lambda$ to
denote a service, $Q_\lambda$ to be the set of requests it serves, and
$t_\lambda$ to denote the time of the service. We say that $\lambda$
\emph{transmits} item $i$ if $i \in v(Q_\lambda)$. A feasible solution is a set
of services $\Lambda$ such that every request $q$ is served at some time in the interval
$[a_q, d_q]$. Denote by $\Lambda_i$ the subset of $\Lambda$ that transmits item
$i$. The cost of $\Lambda$ is
\[\sum_{\lambda \in \Lambda} w_0 + \sum_{i \in v(Q_\lambda)} w_i =
  |\Lambda| w_0 + \sum_i w_i|\Lambda_i|.\] We call $|\Lambda|w_0$ the
\emph{joint cost} of $\Lambda$ and $\sum_iw_i|\Lambda_i|$ the \emph{item cost}
of $\Lambda$.

As is common in the JRP literature, we assume that the request deadlines are \emph{distinct}. This is without loss of generality as one can convert an instance to have this property by fixing an arbitrary permutation of the items and applying infinitesimal perturbations to the deadlines according to that permutation.\footnote{Note that different permutations may result in a different number of item inversions, but our analysis would apply to whichever permutation is being considered.}  (Our lower-bound examples use repeated deadlines for clarity, but it easy to perturb them to have distinct deadlines.) 

\paragraph{Clairvoyant vs Nonclairvoyant.} In the \emph{clairvoyant} setting,
the online algorithm is given the deadline $d_q$ of request $q$ at its arrival
time. 

In the \emph{nonclairvoyant} setting, the online algorithm is not given
$d_q$ when $q$ arrives; instead, it is only given $d_q$ at time $d_q$.  Notably for our lower bound, we focus on the \defn{limited information} setting, where $d_q$ is only revealed if the request is still pending at that time (i.e., it is not served at any time $t < d_q$).

\subsection*{Deadline Predictions}
This paper considers the prediction variant of JRP-D. As in the nonclairvoyant setting, the actual deadlines $d_q$ are not known when a request $q$ arrives. Instead, a prediction $\predd_q$ is provided.  As in the nonclairvoyant setting, the limited information model means that the true deadline $d_q$ is revealed at time $d_q$ only if the request is still pending at that time. 

As discussed in \cref{sec:intro}, the primary way in which we quantify error is with respect to inversions, most notably instantaneous item inversions. 
\begin{definition}[Request Inversion]
  Requests $q$ and $q'$ are said to be \defn{inverted} if they are on different items, i.e., $v_q \neq v_{q'}$, and their deadlines and predicted deadlines have different orderings, i.e., $d_q <
  d_{q'}$ and $\predd_q > \predd_{q'}$.  Each pair of
  requests that are inverted is called a \defn{request inversion}. We use $R$
  to denote the set of inverted request pairs, and $\numreqinv$ to denote the
  total number of request inversions $|R|$ if there are any, or 1 if there are none.
\end{definition}

\begin{definition}[Item Inversion]
  Items $i$ and $j$, $i\neq j$, are said to be \defn{inverted} if there exist requests $q$
  on item $i$ and $q'$ on item $j$ that are inverted. Each pair of items that are inverted is called an \defn{item inversion}. We use $\numiteminv$ to denote the total number of item inversions if there are any, or 1 if there are none.
\end{definition}

\begin{definition}[Instantaneous Inversions]
  The \defn{instantaneous request inversions at time $t$} is set of inversions
  from requests that overlap time $t$, i.e., $R_t = \{(q,q') \in R: t \in [a_q,d_q]
  \cap [a_{q'},d_{q'}] \}$.  The
  \defn{number of instantaneous request inversions} refers to the maximum at any
  time, i.e., $\max_t |R_t|$.

  Similarly, the \defn{instantaneous item inversions at time $t$} refers to any item inversions that overlap at time $t$, i.e.,
  $I_t = \{(v_{q},v_{q'}) : v_{q}\neq v_{q'} \text{ and } (q,q') \in
  R_t\}$. The \defn{number of instantaneous item inversions} refers to the maximum at any time, i.e., $\max_t |I_t|$.
\end{definition}

We will use $\numiteminv$ to denote the number of instantaneous item inversions in an input, with the exception that if there are no such inversions then we will set $\numiteminv = 1$.  We slightly abuse notation by letting $\numiteminv$ be the number of item inversions when we prove our lower bound (Theorem~\ref{thm:lb-main}).  Note that this is a stronger lower bound than if we used instantaneous item inversions.

\section{Main Upper Bound: the \local{} Algorithm}
\label{sec:local-greedy}
\paragraph{Algorithm description.} Our algorithm, which we call \local{}, is given in \cref{alg:stack-greedy}.  The algorithm keeps track of phases, defined as follows. The first phase starts at the first deadline; subsequently, the current phase ends and a new phase starts when a deadline
hits for a request whose interval is disjoint from the request interval that started the phase. When a new phase starts, the service that is also triggered at the start of the new phase belongs to the new phase. Next, we describe how the algorithm chooses the requests to serve in a service. Suppose a service is triggered within a phase that starts at time $s$ when the deadline of an unserved request $q$ is reached. The request $q$ is called the \emph{triggering request} of the service. The only requests that are \emph{eligible} for service are those whose intervals intersect with $s$. The algorithm goes through the eligible requests in ascending order of deadline, adding each request's item to its transmission set $I$, and stops once $I$ served has total item cost $w(I) \geq w_0$, or there is no eligible request left. It then transmits $I$, serving all pending requests on items in $I$. The remaining eligible requests are called \emph{leftover requests}.

\begin{algorithm}
  \algblockdefx[Name]{When}{EndWhen}[1]{\textbf{when} #1:}{\textbf{end when}}
  \caption{\local{} Algorithm}
  \label{alg:stack-greedy}
  \begin{algorithmic}[1]
  \State initialize $s = -\infty$ \Comment $s$ tracks the start time of the current phase
    \When{the deadline of an unserved request $q$ is reached}
    \If{$a_q > s$}
      \State start a new phase and set $s = d_q$
    \EndIf
    \State let $E_\lambda$ be the set of pending requests $q'$ with $a_{q'} \leq s$
    \State initialize $I = \{v_q\}$;
    \For{each $q' \in E_\lambda$ in ascending order of deadline}
      \State add $v_{q'}$ to $I$
      \If{$w(I) \geq w_0$}
        \State break
      \EndIf
    \EndFor
    \State transmit $I$ and serve all pending requests on items $I$
    \EndWhen
  \end{algorithmic}
\end{algorithm}

\paragraph{Overview of analysis.} Let $\Lambda$ be the set of services produced
by \local. 
Let $P$ be the set of requests whose deadlines correspond to the start of a phase. 
Observe that the
optimal solution makes at least $|P|$ services since $P$ is a set of disjoint
requests and so $\OPT \geq w_0|P|$. Next, consider some phase and let $\Lambda'$ be the services within the phase. Since each service of \local{} costs at most
$3w_0$, the above lower bound on $\OPT$ lets us pay for the cost of one service the phase. Observe that each
of the first $|\Lambda'|-1$ services of $\Lambda'$ have item cost at least
$w_0$. This is because a service can only have item cost less than $w_0$ if
there are no requests leftover, and thus, is the last service of the phase. In the clairvoyant setting, the item costs of these services can be charged to the item cost of $\OPT$ because there are sufficiently many disjoint requests on each item. In the nonclairvoyant setting with predicted deadlines, we run \local with predicted deadlines. Here, we need to divide these services into safe and unsafe services. For the safe services, their item cost can be analyzed as before, while the item cost of unsafe services are charged to $w_0|P|$ by showing that there are at most $O(\sqrt{\numiteminv})$ unsafe services per phase.

\subsection{Useful Properties and Non-Uniform Bound} \label{sec:local-greedy-non-uniform}
We first prove some useful properties of the optimal solution and \local{} that hold independent of the accuracy of the predictions.  We will then be able to use these to get our first bound: an $O(\sqrt{\numiteminv})$ competitive ratio.

\paragraph{Lower bounds on $\OPT$.} Consider an optimal solution $\Lambda^*$ with cost $\OPT$. Let $\Lambda^*_i$ be
the subset of $\Lambda^*$ that transmit item $i$. Let $\OPT_0 := w_0|\Lambda^*|$
denote the joint cost of the optimal solution and $\OPT_i := w_i |\Lambda^*_i|$
be its item $i$ cost.

\begin{prop}
  \label{prop:disjoint-lb}
  Let $R$ be a set of disjoint requests and for each item $i$, let $R_i$ be set
  of disjoint requests on item $i$. Then, $\OPT_0 \geq w_0|R|$ and $\OPT_i \geq
  w_i|R_i|$ for each item $i$.
\end{prop}

\begin{proof}
  Since the requests in $R$ are disjoint, the optimal solution must have at
  least $|R|$ services and so $|\Lambda^*| \geq |R|$. Let $\Lambda^*_i$ be the
  set of services in $\OPT$ that transmit item $i$. Since there are $|R_i|$
  disjoint requests on item $i$, the optimal solution must transmit item $i$ at
  least $|R_i|$ times and so $|\Lambda^*_i| \geq |R_i|$.
\end{proof}

\paragraph{Bounding the cost of charged services.} The crux of the analysis of \local is in bounding the cost of charged services, defined as follows.

\begin{definition}[Charged services]
  A service $\lambda$ is a \emph{charged} if it is not the last service of its phase, and \emph{uncharged} otherwise.
\end{definition}

Let $\cservices$ be the set of
charged services. The following proposition implies that to bound the cost of charged services, it suffices to bound their total item cost.

\begin{prop}
  \label{prop:charged}
  Let $\lambda$ be a charged service. Then, $w(\lambda) \geq w_0$.
\end{prop}

\begin{proof}
  Suppose, towards a contradiction, that $w(\lambda) < w_0$. Let $\lambda'$ be
  the next service of the phase and $q'$ be its triggering request. Let $s$ be
  the start of the phase. Since $\lambda$ and $\lambda'$ are in the same phase,
  $q'$ intersects $s$ and so is eligible when $\lambda$ was triggered. Since
  $w(\lambda) < w_0$, \local{} would have added it to $\lambda$, a
  contradiction.
\end{proof}

Let $\ciservices$ be the subset of charged services that transmit item $i$. To bound the total item cost of $\cservices$, we will bound, for each item $i$, the number of charged services transmitting item $i$ in terms of $|\Lambda_i^*|$, the number of times the optimal solution transmits item $i$. To that end, we consider the requests that caused $\local$ to transmit item $i$.

\begin{definition}[$i$-triggering requests]
  For each service $\lambda \in \ciservices$, for each item $i$, let
  $q^i_\lambda$ be the earliest-deadline pending request that caused $\lambda$
  to include item $i$ in its transmission. Let $Q_i := \{q^i_\lambda\}_{\lambda \in \ciservices}$.
  The requests in $Q_i$ are called \emph{$i$-triggering requests}.
\end{definition}

\begin{prop}
  \label{prop:stack}
  During a phase, every item is transmitted at most once.
\end{prop}

\begin{proof}
  Suppose, towards a contradiction, that there a two services $\lambda$ and
  $\lambda'$ in a phase that transmits item $i$. Let $s$ be the start of the
  phase, $t, t'$ be their respective service times, and suppose $t < t'$. Let
  $q'$ be the $i$-triggering request of $\lambda'$. Since $q'$ was served by
  $\lambda'$, we have $a_{q'} < s$ and $d_{q'} \geq t' > t$. Thus, when
  $\lambda$ transmitted item $i$, it would have served $q'$ as well. However,
  this contradicts the assumption that $q'$ was served by $\lambda'$.
\end{proof}

\paragraph{Bounds in terms of unsafe services.} We are now ready to develop the tools needed to relate the competitive ratio of \local with the prediction error. The main tool is the following notion of safe and unsafe services, which capture the intuition that wrong predictions hurt \local by causing it to serve a request too early.

\begin{definition}[Phase boundary and safe requests/services]
  The \emph{boundary} of a phase is the time of its last service. A request is   \emph{unsafe} if its deadline is later than the boundary of the phase it
 is served in and \emph{safe} otherwise. A charged
  service is \emph{safe} if every request $q$ in the service is safe, and \emph{unsafe} otherwise.
\end{definition}

 We will later show that the item cost of safe services can be charged to the item cost of the optimal solution (\cref{lem:safe}), and prove a bound on the competitive ratio of \local in terms of unsafe services (\cref{lem:unsafe}). Together, they show that a set of predictions is bad for \local, in the sense that they cause \local to have large competitive ratio, only if they cause \local to have many unsafe services in a single phase.

\begin{lemma}
  \label{lem:safe}
  Let $\Lambda'_i$ be a set of safe services that transmit item $i$. Then,
  $w_i|\Lambda'_i| \leq 2\OPT_i$.
\end{lemma}

\begin{proof}
  Let $Q'_i$ be the set of $i$-triggering requests that caused the algorithm to
  add item $i$ to the services in $\Lambda'_i$. We will show that the depth of
  $Q'$ is at most $2$, which in turn implies that there are at least
  $|\Lambda'_i|/2$ disjoint requests on item $i$.
  
  Suppose, towards a contradiction, that requests $q_1, q_2, q_3 \in Q'_i$
  overlap. Suppose that these requests were served by safe services $\lambda_1,
  \lambda_2, \lambda_3$ at times $t_1 < t_2 < t_3$, respectively. By
  \cref{prop:stack}, the requests were served in different phases.
  Let $s_1, s_2, s_3$ be the start times of the respective phases, and $b_1,
  b_2, b_3$ be the boundaries of the respective phases. Note that $b_1 < s_2
  \leq t_2 \leq b_2 < s_3 \leq t_3$. Since $\lambda_1$ served $q_1$ and
  $\lambda_1$ is a safe service, the deadline of $q_1$ is earlier than $b_1$. On
  the other hand, $q_3$ has deadline no earlier than $t_3$, since it was served
  at that time. Thus, it must have arrived after $t_2$, as otherwise it would
  already have been served at $t_2$. Therefore, $q_1$ and $q_3$ do not overlap,
  a contradiction.

  Thus, there are at least $|\Lambda'_i|/2$ disjoint requests on item $i$.
  Combining this with \cref{prop:disjoint-lb} gives us the lemma.
\end{proof}

\begin{lemma}
  \label{lem:unsafe}
  Suppose that there are at most $D$ unsafe services per phase. Then, the
   cost of \local is at most $(2D+3)\OPT_0 + 4\sum_i \OPT_i \leq (2D+4)\OPT$.
\end{lemma}

\begin{proof}
Recall that $P$ is the set of disjoint requests whose deadlines mark the start of a phase. We show that the cost of uncharged services $\Lambda \setminus \cservices$ is at most $3\OPT_0$ and
the cost of charged services is at most $2D\OPT_0 + 4\sum_i \OPT_i$.
  
By definition of \greedy, the cost of every service $\lambda$ satisfies $w_0 \leq c(\lambda) \leq 3w_0$. 
  So, the cost of each uncharged
  service is at most $3w_0$. Since each phase has at most one uncharged service and $|P|$ is exactly the number of phases, \cref{prop:disjoint-lb} implies 
  \[c(\Lambda \setminus \cservices) \leq 3\OPT_0.\]

  Consider a charged service $\lambda \in \cservices$. Since it is charged, we
  have that $w(\lambda) \geq w_0$ and so $c(\lambda) = w_0 + w(\lambda)\leq
  2w(\lambda)$. Let $\Lambda'$ be the subset of $\cservices$ that is safe, and
  $\Lambda'_i$ be the subset of $\ciservices$ that is safe. The total cost of
  charged services is at most
  \begin{align*}
    \sum_{\lambda \in \cservices} 2 \sum_{i \in v(Q_\lambda)}w_i
    &= 2 \sum_i w_i |\ciservices| \\
    &\leq 2 \sum_i w_i |\Lambda'_i| + 2 \sum_i w_i|\ciservices \setminus \Lambda'_i|\\
    &\leq 4 \sum_i \OPT_i + 2w_0|\cservices \setminus \Lambda'| \tag{\cref{lem:safe} and $w_0 \geq w_i$ for every item $i$}\\   
    &\leq 4 \sum_i \OPT_i + 2Dw_0|P| \tag{at most $D$ unsafe services per phase}\\
    &\leq 4 \sum_i \OPT_i + 2D\OPT_0. \tag{$P$ is disjoint and \cref{prop:disjoint-lb}}
  \end{align*} 

  Together with the above bound on the cost of uncharged services, we  have that the total cost of \local is at most
  \[(2D+3)\OPT_0 + 4\sum_i \OPT_i \leq 2D\OPT_0 + 4 \OPT \leq (2D+4)\OPT.\] This completes the proof of the lemma.
\end{proof}

\paragraph{Bounding unsafe services in terms of prediction error.} One nice property of \local{} is that it has good consistency: if all predictions are accurate, then its competitive ratio is quite good. We remark that we have not attempted to optimize the constants in both the design and analysis of \local.

\begin{theorem}
  \label{thm:stack-clairvoyant}
  In the clairvoyant setting---all predictions are accurate---the competitive ratio of \local is at most $4$.
\end{theorem}

\begin{proof}
  We claim that in the clairvoyant setting, every charged service is
  safe. Consider a phase starting at time $s$ and has boundary $b$. Let $q'$ be
  the triggering request of the last service $\lambda'$ of the phase,
  i.e.~$d_{q'} = b$. Since $\lambda'$ is part of the phase, $a_{q'} \leq s \leq
  b$. Let $\lambda$ be a charged service of the phase. It has service time $t$
  where $s \leq t \leq b$. Thus, $q'$ was eligible for $\lambda$. Since $q'$ was
  not served by $\lambda$, it must have deadline no earlier than those served by
  $\lambda$. Therefore, $\lambda$ is safe. This concludes the proof of the claim. The claim implies that we can apply \cref{lem:unsafe} with $D = 0$ which yields the desired bound on the competitive ratio of \local.
\end{proof}

We can now use the previous lemmas and properties to prove our main upper bound on \local{} (although not our main upper bound, which we get later in \cref{sec:local-greedy-non-uniform} by bucketing items by weight and running \local{} in each bucket).

\begin{theorem}
  \label{thm:stack-item-inv}
  In the predictions setting, the competitive ratio of \local{} is at most
  $O(\sqrt{\numiteminv})$.
\end{theorem}

\begin{proof}
  It suffices to show that if a phase has $D$ unsafe services, then there are at
  least $\Omega(D^2)$ instantaneous item inversions. Consider a phase with $D$
  unsafe services that starts at time $s$ and has boundary $b$. Let $\lambda_1,
  \ldots, \lambda_D$ be the $D$ unsafe services and $t_1 < t_2 < \ldots < t_D$
  be their respective service times. For $\lambda_j$, let $q_j$ be the request served by $\lambda_j$ whose deadline is after the phase boundary, and let $q'_j$ be the leftover request that
  triggered the next service immediately after $\lambda_j$. For convenience, we write $d_j = d_{q_j}$ and $d'_j := d'_{q_j}$. 
  
  We claim that there are $\Omega(D^2)$ inversions among the requests $q_1, q'_1, \ldots, q_D, q'_D$. Observe that by definition of $d'_j$, we get
  \[t_1 < d'_1 \leq t_2 < d'_2 \leq \ldots < t_{D-1} < d'_{D-1} \leq t_D <
    d'_D.\]
  By definition of $d_j$, we also have $d_j > b > d'_k$ for every $j, k \in [1,
  D]$. Next, we re-index as follows: let $q'_{(j)}$ be the leftover triggering
  request with the $j$-th earliest predicted deadline, $\lambda_{(j)}$ be the
  charged service it is leftover from, and $q_{(j)}$ be the request served by
  $\lambda_{(j)}$ whose deadline is after the phase boundary, i.e.~, i.e.~$d_{q_j} > b$. Since \local
   chose to add $q_{(j)}$ to $\lambda_{(j)}$ instead of $q'_{(j)}$, we
  get that $\predd_{(j)} < \predd'_{(j)}$. Thus, for $j \leq k \leq D$,
  $\predd_{(j)} \leq \predd'_{(k)}$ and so $q_{(j)}$ is inverted with
  $q'_{(k)}$. Therefore, we conclude that the total number of inversions among $q_1, q'_1, \ldots, q_D, q'_D$
  is at least $\sum_{j = 1}^{D} (D - j + 1) = \Omega(D^2)$.
  
Since $q_1, q'_1, \ldots, q_D, q'_D$ are eligible requests of their respective services, their intervals contain the start time of the phase. Thus, these inversions are instantaneous. Moreover, each of them is on a distinct item, by \cref{prop:stack}. Therefore, there are $\Omega(D^2)$ instantaneous item inversions and so $D = O(\sqrt{\eta})$.  Applying \cref{lem:unsafe} completes the proof of the lemma.
\end{proof}

\subsubsection{Tightness of Analysis} \label{sec:local-greedy-tight}
We now show that our analysis above is tight (as discussed in \cref{sec:intro}).
\begin{theorem} \label{thm:local-greedy-tight}
    There is an instance on which \local{} has competitive ratio $\Omega(\sqrt{\numiteminv})$.
\end{theorem}
\begin{proof}
    Let $w_0 = 1$.  Our instance will have $2n$ items: cheap items $C = \{c_1, \dots, c_n\}$ together with expensive items $E = \{e_1, \dots, e_n\}$.  Each $c_i$ has item ordering cost $1/n$, and each $e_i$ has item ordering cost $1$.  

    Our requests come in $n$ phases, where phase $i$ starts at time $t_i = 2n(i-1)$ and has length $2n$.  These phases will correspond with the phases of \local{}.  In phase $i$, a request for each item is released at time $t_i$.  The request for an expensive item $e_j$ has true deadline $3n^2$ (i.e., after all the phases have finished) and predicted deadline $t_i + 2(j-1) + 1$.  The request for a cheap item $c_j$ has true and predicted deadline $t_i + 2(j-1)$.  Note that $c_j$ is inverted with $e_{j'}$ for all $j' < j$, and hence there are $\Theta(n^2)$ item inversions.  

    Since all of the expensive item requests have true deadline at $3n^2$, we can use one service at that time to handle all of those requests (with cost $n+1$), and then in each phase use one service to handle all of the cheap item requests (so each such service costs $1 + n(1/n) = 2$).  Thus $\OPT \leq 2n + n+1 = 3n+1$.

    On the other hand, in each phase when \local{} hits a deadline (for a cheap item $c_j$) it will do a service for $\{c_j, e_j\}$ with cost $2 + \frac1n$.  Hence the total cost of \local{} on this instance is $n(2+\frac1n)$ per phase, and so $\Theta(n^2)$ overall.  Thus the competitive ratio is $\Omega(n) = \Omega(\sqrt{\numiteminv})$ as claimed.
\end{proof}

\subsection{Refined Analysis for Uniform Weights}

To bypass the lower bound of \cref{thm:local-greedy-tight}, we show that when run on uniform weights, \local{} actually performs better. In particular, we prove the following theorem.

\begin{theorem}
  \label{thm:local-uniform}
  \local{} is $O(\sqrt[3]{\numiteminv})$-competitive when every item has
  ordering cost $0 \leq w \leq 1$.
\end{theorem}

\paragraph{Overview.} The main idea for the improved analysis is in observing that \cref{lem:unsafe} charges much more to $\OPT_0$, joint cost of the optimal solution, than to the item ordering cost, $\sum_i \OPT_i$. To fix this, we use a more refined notion of unsafe services (\cref{def:t-unsafe}). Roughly speaking, instead of saying the entire service is unsafe and charging its entire cost to $\OPT_0$ if it has even a single unsafe request, we only do so if there are many unsafe requests. Together with a bound on \local (\cref{prop:weight-bound}) in terms of the uniform item ordering cost $w$, we get the improved analysis.

\begin{prop}
  \label{prop:weight-bound}
  \local{} is $O(1/w)$-competitive.
\end{prop}

\begin{proof} 
  We begin by observing that the cost of each service of \local{} is at most
  $2$ and so the total cost of \local{} is at most $2|\services|$. Let
  $\tiservices$ be the set of services triggered by a request on item $i$. Since
  the triggering requests on the same item are disjoint, we get
  $\OPT_i \geq w|\tiservices|$ for each item $i$. Since $|\services| = \sum_i |\tiservices|$, we now have that
  \[|\services| = \sum_i |\tiservices| \leq \frac{1}{w} \sum_i\OPT_i,\]
  as desired. 
\end{proof}

\begin{definition}[$\tau$-unsafe services]
\label{def:t-unsafe}
 A charged service is \emph{$\tau$-unsafe} if at least $\tau$ fraction of its requests
 are unsafe and \emph{$\tau$-safe} otherwise.
\end{definition}

\begin{lemma}
  \label{lem:t-unsafe}
  For every $0 < \tau < 1$, the cost of \local{}'s solution is at most \wunote{I changed $t$ to $\tau$ because we use $t_i$ to mean service times later on and thought it may be confusing. Please check that I've made the necessary substitutions below}
  \begin{align*}
    \sqrt{\frac{w}{\tau}\numiteminv}\cdot\OPT_0 + \frac{1}{1-\tau}\sum_i \OPT_i
  \end{align*}
  and so \local{} is $O(\max\{\sqrt{\frac{w}{\tau}\numiteminv},
  \frac{1}{1-\tau}\})$-competitive.
\end{lemma}

\begin{proof}
We will show that the total cost of $\tau$-unsafe services is at most
$\sqrt{\frac{w}{\tau}\numiteminv}\cdot\OPT_0$ and the total cost of $\tau$-safe
services is at most $\frac{1}{1-\tau}\sum_i \OPT_i$.

Let $p$ be the number of phases. Since the triggering requests of the first
services of the phases are disjoint, we get $\OPT_0 \geq p$. We now bound the
number of $\tau$-unsafe services per phase. 

\begin{claim}
  Every phase has at most $O(\sqrt{\frac{w}{\tau}\numiteminv})$ $\tau$-unsafe services.
\end{claim}

\begin{proof}
  We prove the contrapositive: 
  If some phase has $D$ $\tau$-unsafe services, then $\numiteminv
  \geq \Omega(D^2\tau/w)$.
  
  Consider a phase with $D$ $\tau$-unsafe services that starts at time $s$ and has
  boundary $b$. Let $\lambda_1, \ldots, \lambda_D$ be the $D$ unsafe services
  and $t_1 < t_2 < \ldots < t_D$ be their respective service times. For
  $\lambda_j$, let $Q_j \subseteq Q(\lambda_j)$ be consist of requests $q$ with
  deadline $d_q > b$ and let $q'_j$ be the leftover request that triggered the
  next service after $\lambda_j$. Thus, we have
  \[t_1 < d'_1 \leq t_2 < d'_2 \leq \ldots < t_{D-1} < d'_{D-1} \leq t_D <
    d'_D < b.\]

  We first observe that by definition, $d_q > b > d'_k$ for every $q \in Q_j$
  and $j,k \in [1, D]$. Next, we re-index as follows: let $q'_{(j)}$ be the
  leftover triggering request with the $j$-th earliest predicted deadline,
  $\lambda_{(j)}$ be the charged service it is leftover from, and $Q_{(j)}$ be
  the requests served by $\lambda_{(j)}$ whose deadlines are after the phase
  boundary. Since the greedy algorithm chose to add $Q_{(j)}$ to $\lambda_{(j)}$
  instead of $q'_{(j)}$, we get that $\predd_{q} < \predd'_{(j)}$ for every $q
  \in Q_{(j)}$. Thus, for $j \leq k \leq D$ and $q \in Q_{(j)}$, we have $\predd_{q} \leq \predd'_{(k)}$
  and so every request in $Q_{(j)}$ is inverted with $q'_{(k)}$. Since
  $|Q_{(j)}| \geq \tau/w$, we conclude that the
  total number of item inversions is at least $\sum_{j = 1}^{D} (k - j + 1)\tau/w =
  \Omega(D^2\tau/w)$.
\end{proof}

This implies that each phase has at most $O(\sqrt{\frac{w}{\tau}\numiteminv})$
$\tau$-unsafe services. Thus, the total cost of $\tau$-unsafe services is at most $p
\cdot 2\sqrt{\frac{w}{\tau}\numiteminv} \leq
2\sqrt{\frac{w}{\tau}\numiteminv}\cdot \OPT_0$. This completes the first half of
the proof.

It remains to bound the total cost of $\tau$-safe services. The cost of a $\tau$-safe
service is at most $2$ and so is at most $2/(1-\tau)$ times the item cost of its
safe requests. More precisely, let $Q'_i$ be the set of $i$-triggering requests
that are safe. Then, the total item cost of safe requests is $\sum_i w|Q'_i|$
and we have that the total cost of $t$-safe services is at most $O\left(\frac{1}{1-\tau}
\sum_i w|Q'_i|\right)$.

\begin{claim}
For every item $i$, the item ordering cost of \local due to item $i$ is $w|Q'_i|\leq 2\OPT_i$.
\end{claim}
\begin{proof}
  Let $Q'_i$ be the set of $i$-triggering requests. We will show that the depth of
  $Q'_i$ is at most $2$, which in turn implies that there are at least
  $|Q'_i|/2$ disjoint requests on item $i$.

  Suppose, towards a contradiction, that requests $q_1, q_2, q_3 \in Q'_i$
  overlap. Suppose that these requests were served by services $\lambda_1,
  \lambda_2, \lambda_3$ at times $t_1 < t_2 < t_3$, respectively. By
  \cref{prop:stack}, the requests were served in different phases.
  Let $s_1, s_2, s_3$ be the start times of the respective phases, and $b_1,
  b_2, b_3$ be the boundaries of the respective phases. Note that $b_1 < s_2
  \leq t_2 \leq b_2 < s_3 \leq t_3$. Since $\lambda_1$ served $q_1$ and
  $q_1$ is a safe request, the deadline of $q_1$ is earlier than $b_1$. On
  the other hand, $q_3$ has deadline no earlier than $t_3$, since it was served
  at that time. Thus, it must have arrived after $t_2$, as otherwise it would
  already have been served at $t_2$. Therefore, $q_1$ and $q_3$ do not overlap,
  a contradiction.

  Thus, there are at least $|Q'_i|/2$ disjoint requests on item $i$.
  Combining this with \cref{prop:disjoint-lb} gives us the claim.
\end{proof}
Thus, the total cost of $\tau$-safe services is at most $O\left(\frac{1}{1-\tau}
\sum_i w|Q'_i|\right) \leq O\left(\frac{1}{1-\tau}
\sum_i \OPT_i\right)$. This completes the second half of the proof.
\end{proof}

\begin{proofof}{\cref{thm:local-uniform}}
  Applying \cref{lem:t-unsafe} with $\tau = 1 - 1/\sqrt{w\cdot \numiteminv} \geq
  1/2$, we get that \local{} is $O\left(\sqrt{w\cdot \numiteminv}\right)$-competitive. By
  \cref{prop:weight-bound}, we get that the competitive ratio of \local{} is at
  most
  \[O(\min\{\sqrt{w\cdot\numiteminv}, 1/w\}) = O(\sqrt[3]{\numiteminv}),\] as
  desired.
\end{proofof}

\subsection{Final Bound}

We can now bucket by weights and apply \cref{thm:local-uniform} to each bucket.  We call this algorithm ``\local{} with bucketing''.  This will give our main theorem, which we restate here for completeness.

\maintheorem*
\begin{proof}
    Our algorithm is a simple variant of \local{}.  For $j \in [\log n]$, let $B_j = \{i : w_0 / 2^i < w_i \leq w_0 / 2^{i-1}\}$.  Let $B_{\log n + 1} = \{i : w_i \leq w_0 / n\}$.  For each bucket $B_i$ with $i \leq \log n$, we create a modified instance where all ordering costs have been rounded up to $w_0 / 2^{i-1}$.  We then run \local{} independently on the modified instance for each bucket, except for bucket $B_{\log n + 1}$ we run the trivial algorithm that always services all open requests whenever a deadline is reached.  
    
    To analyze this algorithm, we need to set up some notation.  For each $i \in [\log n +1]$, let $\ALG_i$ denote the cost of the algorithm on the jobs in bucket $B_i$.  Let $\ALG = \sum_{i=1}^{\log n +1} \ALG_i$ be the total cost of our algorithm.  Let $\numiteminv_i$ be the number of instantaneous item inversions in the instance restricted to the items in $B_i$, and note that $\numiteminv_i \leq \numiteminv$ for all $i$.  Let $\OPT_i$ denote the optimal cost of the original instance restricted to items in $B_i$, and let $\OPT'_i$ denote the optimal cost of the modified instance restricted to items in $B_i$.  Clearly $\OPT_i \leq \OPT$ for all $i$, and also $\OPT'_i \leq 2 \OPT_i$ for all $i \leq \log n$.  

    For bucket $B_{\log n + 1}$, we know by the definition of the ordering costs that every time we perform a service that total cost is at most $w_0 + n \frac{ w_0}{n} \leq 2w_0$.  Let $\alpha$ denote the total number of times the algorithm performs a service.  We know that the triggering requests for any two services must correspond to disjoint intervals, since otherwise by the definition of the algorithm they would both be served at the same time, and hence $\OPT \geq w_0 \alpha$.  Thus 
    \begin{equation} \label{eq:last-bucket}
        \ALG_{\log n + 1} \leq 2w_0 \alpha \leq 2\OPT.
    \end{equation}
    
    Putting these together we have that
    \begin{align*}
        \ALG &= \sum_{i=1}^{\log n} \ALG_i + \ALG_{\log n + 1} \\
        &\leq \sum_{i=1}^{\log n} \ALG_i + 2\OPT \tag{Eq.~\eqref{eq:last-bucket}} \\
        &\leq \sum_{i=1}^{\log n} O(\numiteminv_i^{1/3}) \cdot \OPT'_i + 2\OPT \tag{\cref{thm:local-uniform}} \\
        &\leq \sum_{i=1}^{\log n} O(\numiteminv_i^{1/3}) \cdot \OPT + 2\OPT \\
        &\leq \left(\sum_{i=1}^{\log n} O(\numiteminv_i)\right)^{1/3} \log^{2/3} n \cdot \OPT + 2\OPT \tag{H\"older's inequality}\\
        &\leq O(\numiteminv^{1/3} \log^{2/3} n) \cdot \OPT
    \end{align*}
    as claimed.
\end{proof}

\section{Combining Algorithms}\label{sec:combine}

We now show that it is possible to combine different algorithms to get the ``best of all worlds'' guarantee, thereby obtaining Theorem~\ref{thm:maincombined}.  This is a relatively standard goal in online algorithms, both with predictions and a more general case of combining different algorithms where we do not know ahead of time which will be better (see, e.g., \cite{MNS12}).  The standard technique for this is to simply use a switching argument: if we want to combine two algorithms then we start running one of them, but keep track of how much the other algorithm \emph{would} have cost if we had run it on the same sequence of requests.  If we notice that the other algorithm is significantly cheaper, then we switch.  Unfortunately we cannot apply this idea to our setting, since when we do a service with some set of requests then we will never find out the true deadlines for any requests other than the triggering requests, and so cannot keep track of the counterfactual cost.  So instead of switching between algorithms, we instead simply take the \emph{union} of the algorithms.  

\begin{definition}[Combinable Algorithm]\label{def:combinable}
We say that an algorithm $\A$ for JRP-D is \emph{combinable} if it satisfies the following three properties.
\begin{enumerate}
    \item \label{def:combinable-service}$\A$ only does a service when an active request hits its deadline.
    \item \label{def:combinable-cost} Whenever $\A$ does a service, its cost is $O(w_0)$.
    \item \label{def:cominable-removal} Consider some set of requests $Q$.  Let $Q' \subseteq Q$.  Then $\A$ behaves the same (makes the exact same services of the exact same requests) when run only on $Q'$ as it does if it is run on $Q$ except the requests in $Q \setminus Q'$ are not serviced by $\A$ but are instead removed at no cost at some arbitrary time between their release dates and deadlines.
\end{enumerate}
\end{definition}

We will now show that it is possible to combine algorithms that are combinable to get the stronger of the two.
\begin{theorem} \label{thm:combine}
    Let $\A_1$ and $\A_2$ be combinable algorithms for JRP-D with predictions with competitive ratios of $\gamma_1(\numiteminv, n)$ and $\gamma_2(\numiteminv, n)$, respectively, where each $\gamma_i$ is a nondecreasing function of $\numiteminv$ and $n$.  Then there is an algorithm with competitive ratio $O(\min(\gamma_1(\numiteminv, n), \gamma_2(\numiteminv, n)))$.
\end{theorem}
\begin{proof}
    Consider the following algorithm $\A$, which essentially just runs both $\A_1$ and $\A_2$ in parallel.
    \begin{itemize}
        \item Initialize both $\A_1$ and $\A_2$ as usual.
        \item When a deadline for an active request hits, let $S_1$ be the service (set of items) that $\A_1$ would do and let $S_2$ be the service that $\A_2$ would do.  Service $S_1 \cup S_2$.
    \end{itemize}
    Note that this algorithm is well-defined thanks to property~\ref{def:combinable-service} of \cref{def:combinable}. 

    To analyze this algorithm, consider some instance (set of requests) $Q$.  For any $R \subseteq Q$, let $A_i(R \mid Q)$ be the cost of running $\A_i$ on input $Q$ but where the requests in $Q \setminus R$ are removed at no cost at some point between their release dates and deadlines, and let $\alpha_i(R \mid Q)$ denote the \emph{number} of services that $\A_i$ does in the same setting.  From property~\ref{def:combinable-cost} of \cref{def:combinable} we know that $A_i(R \mid Q) = O(w_0 \alpha_i(R \mid Q))$.

    Now let $Q_1 \subseteq Q$ be the requests serviced by $\A_1$ when we run $\A$ on $Q$, and let $Q_2 \subseteq Q$ be the requests serviced by $\A_2$ (note that during a service both algorithms could choose to service overlapping requests, and in fact will always both choose to service the request triggering the service, so these are not disjoint).  We also know that $\alpha_1(Q_1 \mid Q) = \alpha_2(Q_2 \mid Q)$, since $\A$ will do a service exactly when an active request hits its deadline (property~\ref{def:combinable-service} of \cref{def:combinable}) which implies that this triggering request is in both $Q_1$ and $Q_2$ and will contribute $1$ to both $\alpha_1(Q_1 \mid Q)$ and $\alpha_2(Q_1 \mid Q)$.  So let $\alpha = \alpha_1(Q_1 \mid Q) = \alpha_2(Q_2 \mid Q)$.  Moreover, by property~\ref{def:cominable-removal} of \cref{def:combinable} we know that $A_i(Q_i \mid Q) = A_i(Q_i \mid Q_i)$ and $\alpha_i(Q_i \mid Q) = \alpha_i(Q_i \mid Q_i)$.
    
    Thus we have that
    \begin{align*}
        A(Q) &\leq A_1(Q_1 \mid Q) + A_2(Q_2 \mid Q) \tag{definition of $\A$} \\
        & \leq O(w_0 \alpha_1(Q_1 \mid Q)) + O(w_0 \alpha_2(Q_2 \mid Q)) \tag{property~\ref{def:combinable-cost} of \cref{def:combinable}} \\
        &= O( w_0 \alpha ) \\
        &= O(\min( A_1( Q_1 \mid Q), A_2(Q_2 \mid Q))) \\
        &= O(\min(A_1(Q_1 \mid Q_1), A_2(Q_2 \mid Q_2))) \tag{property~\ref{def:cominable-removal} of \cref{def:combinable}} \\
        &\leq O(\min(\gamma_1(\numiteminv(Q_1), |Q_1|) \cdot \OPT(Q_1), \gamma_2(\numiteminv(Q_2), |Q_2|) \cdot \OPT(Q_2))) \tag{def of $\gamma_1, \gamma_2$} \\
        &\leq O(\OPT(Q) \cdot \min(\gamma_1(\numiteminv(Q_1), |Q_1|), \gamma_2(\numiteminv(Q_2), |Q_2|))) \tag{monotonicity of $\OPT$} \\
        &\leq O(\OPT(Q) \cdot \min(\gamma_1(\numiteminv, n), \gamma_2(\numiteminv,n))) \tag{monotonicity of $\numiteminv$ and $\gamma_i$}
    \end{align*}
    as claimed.
\end{proof}

Now that we know that we can combine algorithms satisfying \cref{def:combinable}, we need to show that the algorithms we consider are in fact combinable.

\begin{lemma} \label{lem:local-combinable}
    Both variants of \local{}, i.e., with or without bucketing, are combinable. 
\end{lemma}
\begin{proof} 
    \local{} obviously satisfies properties~\ref{def:combinable-service} and \ref{def:combinable-cost} of \cref{def:combinable}, by construction.  We just need to show property~\ref{def:cominable-removal}.  %
    As it progresses, \local{} only makes decisions depending on (1) the start time $s$ of the current phase, and (2) the set of pending requests that were released before time $s$. A phase boundary always starts at a deadline for a pending request that was released after $s$, which must  by definition be in $Q'$. Thus, the phase boundaries behave as though any removed requests never existed.  Similarly, when \local{} chooses request to include in a service, it has a total priority order that depends only on the set of still pending requests (i.e., in $Q'$) requests that arrived before~$s$. 
\end{proof}

\subsection{A Combinable nonclairvoyant Algorithm}

In order to make our main result robust, as required by Theorem~\ref{thm:maincombined}, we also need a combinable algorithm that is $O(\sqrt{n})$-competitive, where $n$ is the number if item types.  Unfortunately, the existing nonclairvoyant algorithm~\cite{LeUX23} without predictions does not directly satisfy property~\ref{def:combinable-cost} of \cref{def:combinable}.  Here we provided a simplified algorithm for nonclairvoyant JRP-D that is combinable, building on the same tools. For completeness, we also provide a simple argument that it achieves the desired competitive ratio. (The algorithm from~\cite{LeUX23} solves a more general problem of multi-level aggregation, but the ideas are similar.) 

A main component of the algorithm, as in~\cite{LeUX23}, is dividing the items into heavy and light items.  We call an item~$i$ \defn{heavy} if it has item ordering cost $w_i \geq w_0/\sqrt{n}$ and \defn{light} otherwise. The remainder of the algorithm, however, is a bit different in the way light items are handled to make it fit our definition of combinable. 

\begin{algorithm}
  \algblockdefx[Name]{When}{EndWhen}[1]{\textbf{when} #1:}{\textbf{end when}}
  \caption{Nonclairvoyant JRP-D}
  \label{alg:nonclairvoyant}
  \begin{algorithmic}[1]
    \State partition the $\leq n$ light items arbitrarily into groups $G_1,\ldots,G_{\sqrt{n}}$ such that $|G_j| \leq \sqrt{n}$. 
    \When{the deadline of an unserved request $q$ is reached}
    \State let $i = v_q$ be the corresponding item
    \If{$i$ is heavy, i.e., $w_i \geq w_0 / \sqrt{n}$}
      \State serve all requests for item $i$ 
    \Else 
    \State serve all requests for items in the group $G_j \ni i$
    \EndIf
    \EndWhen
  \end{algorithmic}
\end{algorithm}

Algorithm~\ref{alg:nonclairvoyant} starts by grouping the light items arbitrarily into groups $G_1, G_2, \ldots, G_{\sqrt{n}}$, each with $O(\sqrt{n})$ light items.  Each heavy item and each group are then handled separately.  That is, when a request for heavy item $i$ reaches its deadline, then only that item is served. If instead a request for light item $i$ reaches its deadline, all light items in the same group as $i$ are served. 

\begin{theorem}\label{thm:nonclairv-bound}
    Algorithm~\ref{alg:nonclairvoyant} is $O(\sqrt{n})$ competitive for nonclairvoyant JRP-D.     
\end{theorem}
\begin{proof}
    We start by considering a heavy item~$i$. Let $R_i$ be the set of triggering requests on that item, i.e., the requests that caused $i$ to perform a service.  By construction, the requests in $R_i$ are disjoint as any request that overlaps the deadline $d_q$ of a triggering request $q \in R_i$ are also served at time $d_q$. Thus, by Proposition~\ref{prop:disjoint-lb}, $\OPT_i \geq w_i |R_i|$.  In contrast, Algorithm~\ref{alg:nonclairvoyant} has a cost of $(w_0+w_i)|R_i|$ to serve these requests.  Since $i$ is heavy,  $(w_0+w_i)|R_i| \leq (2\sqrt{n} w_i) |R_i|= O(\sqrt{n}  \OPT_i)$.  Summing across heavy jobs, we thus get our cost for heavy jobs is at most $\sum_{\text{heavy i}} \sqrt{n} \OPT_i) \leq O(\sqrt{n} \sum_i \OPT_i)$
 
    We next consider each particular group $G_j$ of light items.  Again, by construction, the triggering requests $R$ for this particular group are all disjoint.  Thus, by Proposition~\ref{prop:disjoint-lb}, $\OPT_0 \geq w_0 |R|$. The corresponding services that Algorithm~\ref{alg:nonclairvoyant} performs each include at most $|G_j|\leq \sqrt{n}$ items and thus have a cost of at most $w_0 + \sum_{i \in G_j} w_i \leq w_0 + \sqrt{n} w_i \leq 2w_0$.  So the total cost of these services of group $G_j$ are at most $2w_0 |R| = O(\OPT_0)$.  Summing across all $\sqrt{n}$ groups, the total cost of the services performed for light groups is $O(\sqrt{n} \OPT_0)$. 

    Adding the heavy and light costs gives a total cost of $O(\sqrt{n} (\OPT_0 + \sum_i \OPT_i)) = O(\sqrt{n}\OPT)$.
\end{proof}

\begin{lemma}\label{lem:nonclairv-combinable}
Algorithm~\ref{alg:nonclairvoyant} for nonclairvoyant JRP-D is combinable. 
\end{lemma}
\begin{proof}
    Again, each of the properties of \cref{def:combinable} follow directly from the algorithm. Property~\ref{def:combinable-service} is obvious, and Property~\ref{def:combinable-cost} follows from the fact that $w_i \leq w_0$ for heavy items and $\sum_{i \in G_j} w_i \leq w_0$ for groups of light items. Finally, as with \greedy{}, this algorithm's servicing decisions are driven only by the deadline the set of pending request at each particular time, not any history, so removing requests before they are served has no impact.  
\end{proof}

\subsection{Tying it all together}

We are now ready to complete the proof of Theorem~\ref{thm:maincombined}.

\combinedtheorem*
\begin{proof}
    We simply combine three algorithms: unbucketed \local{}, bucketed local{}, and the nonclairvoyant algorithm.
    These algorithms have competitive ratios $O(\sqrt{\numiteminv})$ from~Theorem~\ref{thm:stack-item-inv}, $O(\numiteminv^{1/3}\log^{2/3}(n))$ from Theorem~\ref{thm:main}, and $O(\sqrt{n})$ from Theorem~\ref{thm:nonclairv-bound}, respectively.      Lemmas~\ref{lem:local-combinable} and~\ref{lem:nonclairv-combinable} state that these algorithms are combinable.  We can thus apply Theorem~\ref{thm:combine} to conclude that we achieve the min of all three bounds.  
\end{proof}

\section{Lower Bound}
\label{sec:lb}

Our goal in this section is to prove \cref{thm:lb-main}.  We first give important defintions of the setting.

\begin{definition}
    In the \defn{limited information model}, if a request is serviced before its deadline then the algorithm does not learn its deadline.
\end{definition}

\begin{definition}
    An algorithm is \defn{semi-memoryless} if its memory is empty if there are no outstanding requests, and the algorithm does not use the timestamp when making decisions.
\end{definition}

In other words, if the system is empty (no outstanding requests), then the algorithm is not allowed to remember anything that has happened previously (and in particular cannot remember which items have already been inverted).  Thus if we give the algorithm a sequence of requests which it handles, and then once there are no outstanding requests we give it the exact same sequence of requests, the fact that its memory was empty and it is not allowed to use the time when making decisions means that it will make the exact same decisions on the second copy of the sequence as it did on the first.

Now that the setting has been defined, we can now prove \cref{thm:lb-main}, restated here for convenience.

\mainlb*

\begin{proof}
Consider the following construction, which extends a version of the nonclairvoyant lower bound~\cite{AzarGGP21,LeUX23} to include predictions.  We describe our adversarial construction as adaptive; in reality, because the algorithm is assumed to be deterministic, the adversary could simulate what the algorithm would do before making choices, so it need not be adaptive. 

The goal is to construct an instance on $n$ item types where either (1) the algorithm is at least $c$-competitive when there are no item inversions, or (2) the algorithm is at least $\Omega(\eta^{1/3})$-competitive when $\eta = \Theta(n^{1/3})$.  

We start by setting the joint ordering cost $w_0=1$ and the item ordering costs all to $w_i=1/\sqrt{n}$. 

The construction is broken down into a sequence of $s$ \emph{phases}, where $s$ is to be set later.  The $r$th phase starts at time $2rn$ and ends at time time $2rn+n$.  We now describe each phase, where $t_0 = 2rn$ is the start time. At time $t_0$: for each item type $i\in [n]$, release a request for item $i$ with release time $t_0$ and predicted deadline $\predd_i = t_0+i$. Then advance through time steps $t_0+1,\ldots,t_0+n$ as follows. At time $t_0+i$: if the request for item $i$ has not yet been served by the algorithm then set the actual deadline $d_i = t_0+i = \predd_i$, triggering a service of that request right now.  
Note that thus far we have not set the deadlines of requests that the algorithm chooses to service before their deadlines. To do this, we split the construction into two cases:
\begin{enumerate}
    \item If the algorithm services an average of at most $\sqrt{n} / c$ requests per service (i.e., the total number of services is at least $c \sqrt{n}$), then set the deadline for the request for item $i$ to $d_i = t_0+i = \predd_i$.
    \item If the algorithm services an average of at least $\sqrt{n} / c$ requests per services (so the total number of services is at most $c\sqrt{n}$), then the yet-unset deadlines are instead $d_i = 3sn+t_0+i \gg \predd_i$, where $s$ is the number of phases.   
\end{enumerate}

Observe that by construction, each request is serviced no later than its predicted deadline.  Thus, each phase starts with no ouststanding requests.  Because the algorithm is assumed to be deterministic and semi-memoryless, the algorithm makes exactly the same decisions in each phase. 

To analyze this instance, we break into two cases corresponding to the cases from the construction.  The case depends on the number $x$ of services the algorithm performs during the each phase. 

\paragraph{Case 1.}  Suppose that $x \geq c\sqrt{n}$.  The algorithm thus pays at least $x w_0 + \sum_i w_i \geq c\sqrt{n}\cdot 1 + \sum_i (1/\sqrt{n}) =(c+1)\sqrt{n}$ per phase, for a total of at least $s(c+1)\sqrt{n} \geq sc(\sqrt{n}+1)$ for sufficiently large $n$ (i.e., $\sqrt{n}>c$).  On the other hand, in each phase one could simple service all requests in a single service at the beginning, so the optimal cost is at most $w_0+\sum_i w_i = 1+\sqrt{n}$ per phase or $s(\sqrt{n}+1)$ total.  Thus, the competitive ratio is at least $\frac{sc(\sqrt{n}+1)}{s(\sqrt{n}+1)} = c$.  
Since this case sets all item deadlines to their predicted deadlines, there are no item inversions (i.e., $\eta = 1$), and the competitive ratio of the algorithm is indeed at least $\max\left(\frac{\numiteminv^{1/3}}{c^{1/3}(c+2)}, c\right)$. 

\paragraph{Case 2.} Now suppose that $x < c\sqrt{n}$.  
It follows that most requests are serviced before their predicted deadlines and hence also true deadlines.  As in \cref{sec:local-greedy}, we use the term \emph{triggering requests} to refer to those requests served at their deadlines (i.e., those with $\predd_i = d_i = t_0+i$), and $x$ here equals the number of triggering requests.  Note that all triggering requests have deadlines that fall during the phase, and all non-triggering requests have deadlines after the final phase. 

Since the algorithm services every item in each phase and uses at least one service to do so, the algorithm pays $x w_0 + \sum_i w_i \geq 1+\sqrt{n}$ per phase or at least $s(1+\sqrt{n})$ total.   A better solution would instead serve the $x$ triggering requests at the start of each phase, and those nontriggering requests (with deadlines $d_i = 3sn+t_0+i$) at time $3sn$.  The cost of the optimal algorithm is thus at most $s(w_0 + xw_i) + (w_0 + (n-x)w_i) \leq s(1+c)+(1+\sqrt{n})$.  Putting these together, we get a competitive ratio for the algorithm of at least 
\[
\frac{s(1+\sqrt{n})}{s(c+1)+(1+\sqrt{n})} \geq \frac{s\sqrt{n}}{s(c+1)+s} \text{if $s \geq \sqrt{n}+1$} \ .
\]
Setting $s = n$ thus suffices, giving us a competitive ratio $\rho$ of at least $\rho \geq \frac{s\sqrt{n}}{s(c+2)} = \frac{\sqrt{n}}{c+2}$.  

It remains to relate this competitive ratio to the number of inversions.   By the semi-memoryless property, in every phase the algorithm makes the exact same decisions, and in particular services requests for the exact same items in the exact same order.  So each item either has a corresponding triggering request in every phase, or none of its requests are triggering.  Thus we can partition the \emph{items} into $x$ \emph{triggering items} and $n-x$ \emph{nontriggering items}.  Notice that (1) triggering items are never inverted with triggering items because their predicted deadlines are all correct, and (2) nontriggering items are never inverted with nontriggering items because the $d_i -\predd_i = 3sn$ for every nontriggering request, so order is preserved.  The number of item inversions $\eta$ is thus at most $\eta = x(n-x) \leq xn \leq cn^{3/2}$, where the last step follows from the assumed bound on~$x$.   Taking the cube root, we have $\eta^{1/3} \leq c^{1/3} \sqrt{n}$ and hence $\frac{\eta{1/3}}{c^{1/3}(c+2)} \leq \frac{\sqrt{n}}{c+2} \leq \rho$ as desired, which dominates the maximum when $n$ is large enough relative to~$c$.   
\end{proof}

\section{Conclusions and Future Work}
In this work, we used deadline predictions to bridge the gap between clairvoyant and nonclairvoyant JRP-D. We showed that running the state of the art greedy algorithms \greedy and \mgreedy on the predicted deadlines yields competitive ratios at most $O(\eta)$ and at least $\Omega(\sqrt{\eta})$. To get closer to the lower bound of $\Omega(\eta^{1/4})$, we designed a new greedy algorithm \local that achieves $O(\eta^{1/3})$. Finally, we showed that semi-memoryless algorithms (which includes the above greedy algorithms) have competitive ratios that are $\Omega(\eta^{1/3})$. 

There are several interesting directions for future work. The most immediate
ones are to: close the gap between $O(\eta^{1/3})$ and $\Omega(\eta^{1/4})$;
extend to generalizations of JRP-D such as Cardinality JRP and
Multi-Level Aggregation (see \cite{ChenKU22,Ezra0PRU24} and references therein); and extend to the general delay version. It would also
be interesting to consider nonclairvoyant and learning-augmented algorithms for
the setting with holding costs that penalize early service of a request \cite{MoseleyN025,holdingandbacklogAL,holdingandbacklogSSU}.

\begin{toappendix}
\section{Greedy Algorithm}
\label{sec:greedy}
In this section, we study greedy algorithms as they are particularly natural and one of them is known to be $2$-competitive in the clairvoyant setting \cite{BienkowskiBCJNS14}. %

We begin by defining the greedy algorithms.

\greedy~\cite{BienkowskiBCJNS14} does the
following: when the deadline of a pending request $q$ is reached,
\begin{enumerate}
\item initialize $I = \{v_q\}$;
\item for each pending request $q'$ in ascending order of (predicted) deadline: 
if $w(I \cup \{v_{q'}\}) \geq w_0$, break; else add $v_{q'}$ to $I$; 
\item serve all pending requests on items $I$
\end{enumerate}

For our purposes, it is more convenient to analyze a slightly modified greedy algorithm \mgreedy. \local is also based on \mgreedy. \mgreedy does the
following: when the deadline of a pending request $q$ is reached,
\begin{enumerate}
\item initialize $I = \{v_q\}$;
\item for each pending request $q'$ in ascending order of (predicted) deadline: add $v_{q'}$
  to $I$; if $w(I) \geq w_0$, break
\item serve all pending requests on items $I$
\end{enumerate}

\subsection{Generic properties of \mgreedy}
In this section, we prove properties of \mgreedy that independent of the quality of the predictions.

Let $\Lambda$ be the set of services
produced by \mgreedy.

\begin{definition}[Triggering and leftover requests/items]
  For a service $\lambda \in \Lambda$, the request $q$ whose deadline triggered
  $\lambda$ is called its \emph{triggering request} and the pending requests
  that were not served due to the budget constraint are called \emph{leftover
    requests}. The item $v_q$ of the triggering request is called the
  \emph{triggering item}.
\end{definition}

\begin{definition}[Primary, complete, charged services]
  A service $\lambda$ is called:
  \begin{itemize}
  \item \emph{primary} if its triggering request $q_\lambda$ was not leftover
    from a previous service, and \emph{secondary} otherwise;
  \item \emph{complete} if there are no leftover requests, and \emph{incomplete}
    otherwise;
  \item \emph{charged} if the next service $\lambda'$ was triggered by one of
    its leftover request and say that $\lambda'$ \emph{charges} $\lambda$; the
    service $\lambda$ is \emph{uncharged} otherwise. The triggering request of
    $\lambda'$ is said to to be the \emph{leftover triggering request} of
    $\lambda$
  \end{itemize}
\end{definition}

Let $\pservices$ be the set of primary services. Let $\cservices$ be the set of
charged services and $\ciservices$ be the subset that transmit item $i$. A set
of consecutive services is called a \emph{chain} if the first service is a
primary service and the last service is uncharged. Note that every service is in
a chain. Thus, the number of uncharged services is at most the number of primary
services.

\begin{lemma}
  \label{lem:cost-decomp}
  The cost of \mgreedy is at most
  \[3w_0|\pservices| + 2\sum_i w_i |\ciservices|.\]
\end{lemma}

\begin{proof}
  We show that the cost of uncharged services is at most $w_0|\pservices|$ and
  the cost of charged services is at most $3\sum_i w_i |\ciservices|$.

  By definition of  \mgreedy, the cost of every service $\lambda$
  satisfies $w_0 \leq c(\lambda) \leq 3w_0$. So, the cost of each uncharged
  service is at most $3w_0$. Since the number of uncharged services is equal to
  the number of primary services, we get that
  \[c(\Lambda \setminus \cservices) \leq 3w_0|\pservices|.\]

  Consider a charged service $\lambda \in \cservices$. Since it is charged, we
  have that $w(\lambda) \geq w_0$ and so $c(\lambda) = w_0 + w(\lambda)\leq
  2w(\lambda)$. Thus, the total cost of charged services is at most
  \[\sum_{\lambda \in \cservices} 3 \sum_{i \in v(Q_\lambda)}w_i
    = 2 \sum_i w_i |\ciservices|.\] The lemma now follows from the bounds
  on the cost of uncharged services and the cost of charged services.
\end{proof}

Let $\OPT_0$ denote the cost of the optimal solution due to joint ordering cost
and $\OPT_i$ be its cost due to item $i$ ordering cost.

\begin{lemma}
  \label{lem:primary-disj}
  The set of triggering requests of primary services are disjoint. Thus,
  $w_0|\pservices| \leq \OPT_0$
\end{lemma}

\begin{proof}
  Since each primary service occurs at the deadline of its triggering request,
  and by definition, its triggering request is not leftover from any previous
  service, we get that the set of triggering requests of primary services are
  pairwise-disjoint. Thus, any feasible solution has to have at least
  $|\pservices|$ services. Since each service costs at least $w_0$, we get that
  $w_0|\pservices| \leq \OPT_0$.
\end{proof}

\begin{definition}[$i$-triggering requests]
  For each service $\lambda \in \ciservices$, for each item $i$ except the let
  $q^i_\lambda$ be the earliest-deadline pending request that caused $\lambda$
  to include item $i$. Let $Q_i := \{q^i_\lambda\}_{\lambda \in \ciservices}$.
  We say that the requests in $Q_i$ are \emph{$i$-triggering requests}.
\end{definition}

\begin{lemma}
  \label{lem:item-disj}
  If at most $D$ $i$-triggering requests overlap, then $w_i|\ciservices| \leq D
  \OPT_i$.
\end{lemma}

\begin{proof}
  Since at most $D$ $i$-triggering requests overlap, the set of $i$-triggering
  requests $Q_i$ can be partitioned into at most $D$ disjoint subsets and so
  there exists a disjoint subset of $Q_i$ with at least $|Q_i|/D$ requests.
  Thus, any feasible solution needs to have at least $|\ciservices|/D$ services
  that contain item $i$ and so $w_i|\ciservices| \leq D\OPT_i$.
\end{proof}

\subsection{Analysis of Charged Services}

\begin{definition}[Premature requests]
  Let $q$ be a request served by a charged service $\lambda$ and $q'$ be its
  leftover triggering request. We say that $q$ is \emph{premature} if $d_q >
  d_{q'}$.
\end{definition}

\begin{definition}[Pre-empted requests]
  Let $\lambda$ be a charged service and $q'$ be its leftover triggering
  request. The request $q'$ is said to be \emph{pre-empted} if there exists a
  request $q$ served by $\lambda$ such that $d_q > d'_q$.
\end{definition}

\subsubsection{Request Inversions}

Even though item inversions are a error metric than request inversions (as discussed in \cref{sec:intro}), we get a stronger bound as a function of request inversions, and so if the number of request inversions is not too much larger than the number of item inversions than an analysis based on request inversions it better.  We now give such a bound and analysis.

\begin{lemma}
  \label{lem:req-inv}
  If $D \geq 3$ $i$-triggering requests overlap, then $\numreqinv \geq
  \frac{(D-2)(D-1)}{2}$.
\end{lemma}

\begin{proof}
  Suppose $q_1, \ldots, q_D \in Q_i$ overlap and are served at time $t_1 < t_2 <
  \ldots < t_D$, respectively. For $1 \leq j \leq D$, let $\lambda_j$ be the
  charged service that served $q_j$. Since $\lambda_j$ is charged, it has a
  leftover request $q'_j$ whose deadline triggered the next service $\lambda'_j$
  after $\lambda_j$. For $1 \leq j \leq D$, we write $a_j := a_{q_j}$, $d_j :=
  d_{q_j}$, $a'_j := a_{q'_j}$ and $d'_j := d_{q'_j}$. Recall that, by
  definition, $\lambda'_j$ was triggered at time $d'_j$ and there is no other
  service between $\lambda_j$ and $\lambda'_j$. Thus, we have
  \[t_1 < d'_1 \leq t_2 < d'_2 \leq \ldots < t_{D-1} < d'_{D-1} \leq t_D <
    d'_D.\]

  We first show that $d_j > d'_k$ for every $j, k \in [1, D-2]$. Consider
  request $q_D$. Since it was served at $t_D$, its deadline $d_D \geq t_D >
  t_{D-1}$, so its arrival time $a_D > t_{D-1}$, as otherwise it would have been
  served by $\lambda_{D-1}$. Since $q_j$ overlaps with $q_D$, we have $d_j \geq
  a_D > t_{D-1} > d'_{D-2} \geq d'_k$.

  Next, for $j \in [1, D-2]$, we re-index as follows: let $q'_{(j)}$ be the
  leftover triggering request with the $j$-th earliest predicted deadline,
  $\lambda_{(j)}$ be the charged service it is leftover from, and $q_{(j)}$ be
  the $i$-triggering request served by $\lambda_{(j)}$. Since \mgreedy
   chose to add $q_{(j)}$ to $\lambda_{(j)}$ instead of $q'_{(j)}$, we
  get that $\predd_{(j)} < \predd'_{(j)}$. Thus, for $j \leq k \leq D-2$,
  $\predd_{(j)} \leq \predd'_{(k)}$ and so $q_{(j)}$ is inverted with
  $q'_{(k)}$. Therefore, we conclude that the total number of request inversions
  is at least $\sum_{j = 1}^{D-2} (D-2 - j + 1) = \frac{(D-2)(D-1)}{2}$.
\end{proof}

\begin{theorem}[Request Inversions Bound]
  \label{thm:req-inversions}
  The competitive ratio of \mgreedy is $O(\sqrt{\numreqinv})$.
\end{theorem}

\begin{proof}
  Combining Lemmas~\ref{lem:cost-decomp}, \ref{lem:primary-disj},
  \ref{lem:item-disj} and \ref{lem:req-inv}, we get that the cost of \mgreedy is at most
  \[3w_0|\pservices| + \sum_i w_i |\ciservices| \leq 3\OPT_0 +
    O(\sqrt{\numreqinv})\sum_i\OPT_i \leq O(\sqrt{\numreqinv})\OPT.\qedhere\]
\end{proof}

\subsubsection{Item Inversions}

We now give our main upper bound on \mgreedy: a linear bound in the number of item inversions.  Let $\Prem_i \subseteq Q_i$ be the subset of premature $i$-triggering requests
and $\premlambda_i$ be the services they were served in. The proof of the
following lemma is exactly the same as for \cref{lem:item-disj}.
\begin{lemma}
  \label{lem:nice-disj}
  For every item $i$, no three requests from $Q_i \setminus \Prem_i$ can
  overlap.
\end{lemma}

Let $\inv(i)$ be the set of items that are inverted with $i$. Let $\disj(i)$ be
the maximum number of disjoint requests of item $i$.

\begin{lemma}
  \label{lem:premature}
  $|\premlambda_i| \leq \sum_{j \in \inv(i)} \disj(j)$.
\end{lemma}

\begin{proof}
  Let $\premlambda_{ij}$ be the subset of services $\lambda \in \premlambda_i$
  whose leftover triggering request is on item $j$. By definition,
  $\premlambda_{ij} \neq \emptyset$ only if $j \in \inv(i)$. Since the leftover
  triggering requests on the same item must be disjoint, we get that
  $|\premlambda_{ij}| \leq \disj(j)$. Thus, we get $|\premlambda_i| = \sum_j
  |\premlambda_{ij}| \leq \sum_{j \in \inv(i)} \disj(j)$.
\end{proof}

\begin{theorem}[Item Inversions Bound]
  \label{thm:item-inversions}
  The competitive ratio of \mgreedy is $O(\numiteminv)$.
\end{theorem}

\begin{proof}
  By Lemmas~\ref{lem:cost-decomp} and \ref{lem:primary-disj},
  \ref{lem:item-disj} and \ref{lem:premature}, we get that the cost of the
  greedy algorithm is at most
  \[3w_0|\pservices| + \sum_i w_i |\ciservices| \leq 3\OPT_0 + \sum_i w_i
    |\ciservices|.\]

  We have that
  \begin{align*}
    \sum_i w_i |\ciservices|
    &= \sum_i w_i (|\premlambda_i| + |\ciservices \setminus \premlambda_i|) \\
    &\leq \sum_i w_i\sum_{j \in \inv(i)} \disj(j) + 2\OPT_i \tag{\cref{lem:premature}}\\
    &\leq \numiteminv \cdot \max_j \disj(j) + 2\sum_i\OPT_i \\
    &\leq \numiteminv \cdot \OPT_0 + 2\sum_i\OPT_i.
  \end{align*}
  Thus, the cost of \mgreedy is at most $O(\numiteminv)\OPT$ as claimed.
\end{proof}

\end{toappendix}

\bibliographystyle{alpha} 
\bibliography{references}

\end{document}